\documentclass[%
 reprint,
superscriptaddress,
 amsmath,amssymb,
 aps,
pra
]{revtex4-2}

\usepackage{graphicx}
\usepackage{dcolumn}
\usepackage{bm}
\usepackage{comment} 

\usepackage{amsthm}
\usepackage{mathtools}

\DeclarePairedDelimiterX\set[1]\lbrace\rbrace{#1}

\usepackage{braket}
\usepackage{hyperref}
\usepackage{centernot}

\newtheorem{theorem}{Theorem}
\newtheorem{corollary}{Corollary}[theorem]

\newtheorem{proposition}{Proposition}

\theoremstyle{definition}
\newtheorem{definition}{Definition}

\theoremstyle{remark}
\newtheorem*{remark}{Remark}

\begin{document}

\title{Transport measurements of majorization order for wave coherence}

\author{Cheng Guo}
\email{guocheng@stanford.edu}

\author{David A. B. Miller}

\author{Shanhui Fan}
\email{shanhui@stanford.edu}
\affiliation{
Ginzton Laboratory and Department of Electrical Engineering, Stanford University, Stanford, California 94305, USA
}%

\date{\today}

\begin{abstract}
We investigate the majorization order for comparing wave coherence and reveal its fundamental consequences in transport measurements, including power distribution, absorption, transmission, and reflection. We prove that all these measurements preserve the majorization order under unitary control, enabling direct experimental characterization of the majorization order. Specifically, waves with lower coherence in the majorization order exhibit more restricted ranges of achievable measurement values.  Our results deepen the understanding of coherence in transport phenomena.
\end{abstract}
\maketitle

\clearpage


Wave coherence originates from the statistical properties of random fluctuations~\cite{born1999,goodman2000,mandel1995,oneill2003} and plays a crucial role in fundamental phenomena like interference, diffraction, and scattering~\cite{glauber1963a,mandel1965,perina1985,korotkova2022}. Coherence theory examines how coherence affects observables~\cite{wolf2007}. A fundamental issue in coherence theory is the comparison of coherence between different waves. The concept of ``degree of coherence" can be formalized through various measures, each with specific applications and limitations. For instance, von Laue's entropy measure~\cite{laue1907,gamo1964} has clear thermodynamic significance but is coarse due to its scalar nature~\cite{mandel1965}. Other measures addressing different aspects of coherence were proposed by Zernike~\cite{zernike1938}, Glauber~\cite{glauber1963a}, Mandel and Wolf~\cite{mandel1995}, among others~\cite{gamo1957, karczewski1963a,perina1985,tervo2003,refregier2005,setala2006,luis2007}.

Quantum resource theories have advanced coherence theory~\cite{chitambar2019,torun2023,streltsov2017,baumgratz2014,winter2016}, introducing a new coherence measure based on the majorization order~\cite{nielsen1999,gour2015,luis2016,bengtsson2017,gour2018}. This measure offers finer granularity, clear algebraic and geometric interpretations, and computational simplicity~\cite{marshall2011}. However, its unique physical implications, especially for classical waves, remain unclear. Certainly, any coherence measure, including majorization order, can be indirectly inferred from density matrix tomography~\cite{abouraddy2014,kagalwala2015,roques-carmes2024}. However, direct measurements specific to the majorization order effects are yet to be established. 

In this paper, we reveal the fundamental consequences of the majorization order in transport measurements, including power distribution, absorption, transmission, and reflection. We demonstrate that these measurements, under unitary control (i.e., unitary transformations of the input wave), precisely preserve and manifest the majorization order, distinguishing it from other coherence measures. Consequently, these effects enable direct experimental characterization of the majorization order. Our findings highlight the crucial role of the majorization order in transport phenomena and coherence theory.

We begin by reviewing the density matrix formalism of wave coherence. We consider an $n$-dimensional Hilbert space of waves and focus on the second-order coherence phenomena~\cite{perina1985,wolf2007}. In this formalism, a wave is represented by a density matrix $\rho \in M_n$~\cite{gabor1961,sudarshan1963,gamo1964,gamo1964a,landau1981,oneill2003,wolf2003,delimabernardo2017,zhang2019m,korotkova2022}, also known as a coherence~\cite{wolf1985} or coherency~\cite{goodman2000,yamazoe2012,okoro2017} matrix in optics. Here, $M_n$ denotes the set of $n\times n$ complex matrices. The density matrix $\rho$ is Hermitian and positive semidefinite. A normalized density matrix satisfies:
\begin{equation}\label{eq:normalization_rho}
\operatorname{tr} \rho = 1.
\end{equation}
The coherence properties of the wave are encoded in the eigenvalues of $\rho$, which we call the coherence spectrum:
\begin{equation}
\bm{\lambda}^{\downarrow}(\rho) = (\lambda^{\downarrow}_{1}(\rho),\dots,\lambda^{\downarrow}_{n}(\rho)),
\end{equation}
where $\downarrow$ denotes reordering the components in non-increasing order, which is in principle directly measurable~\cite{roques-carmes2024}. A perfectly coherent wave has $\bm{\lambda}^\downarrow(\rho) = (1,0,\dots,0)$, while a fully incoherent wave has $\bm{\lambda}^\downarrow(\rho) = (1/n, 1/n, \dots, 1/n )$. For any wave, $\bm{\lambda}^\downarrow(\rho)$ belongs to the set of ordered $n$-dimensional probability vectors:
\begin{equation}
\Delta_{n}^\downarrow = \set{ \bm{x} \in \mathbb{R}^{n} \mid x_{i}\geq 0, x_{i} \geq x_{i+1}, \sum_{i=1}^{n} x_{i} = 1 }.
\end{equation}

To compare the coherence of waves, one must introduce an order on $\Delta_{n}^{\downarrow}$. One approach is to use the entropy:
\begin{equation}
H(\rho)=-\sum_{i=1}^{n} \lambda_{i}^{\downarrow}(\rho) \ln \lambda_{i}^{\downarrow}(\rho), 
\end{equation}
and define $\rho_{1}$ to be less coherent than $\rho_{2}$ in the entropy order if $H\left(\rho_{1}\right)>H\left(\rho_{2}\right)$~\cite{laue1907,gamo1964}. We focus on an alternative order based on majorization. For vectors $\bm{x}$ and $\bm{y}$ in $\mathbb{R}^n$, $\bm{x}$ is majorized by $\bm{y}$, denoted as $\bm{x} \prec \bm{y}$~\cite{marshall2011}, if
\begin{align}
&\sum_{i=1}^k x_i^\downarrow \le \sum_{i=1}^k y_i^\downarrow, \quad k=1,2,\ldots,n-1;  \label{eq:def_majorization_ineq}   \\
&\sum_{i=1}^n x_i = \sum_{i=1}^n y_i.
\end{align}
Intuitively, $\bm{x} \prec \bm{y}$ means that the components of $\bm{x}$ are no more spread out than those of $\bm{y}$. The set $\Delta^\downarrow_n$, together with the majorization relation, denoted as $\langle \Delta^{\downarrow}_{n}, \prec\rangle$, forms a partially ordered set~\cite{marshall2011}. For vectors $\bm{x}$ and $\bm{y}$ in $\Delta^\downarrow_n$, if $\bm{x} \prec \bm{y}$ and $\bm{x} \neq \bm{y}$, then $\bm{x}$ is strictly majorized by $\bm{y}$, denoted as $\bm{x} \precneqq \bm{y}$. If neither $\bm{x} \prec \bm{y}$ nor $\bm{y} \prec \bm{x}$ holds, then $\bm{x}$ and $\bm{y}$ are incomparable, denoted as $\bm{x} \parallel \bm{y}$~\cite{davey2002}. Incomparability can occur when $n\ge 3$. (See Appendix, Sec.~\ref{SI-sec:lattice} and~\ref{SI-sec:majorization_lattice} for more details.) Comparing $\rho_1$ and $\rho_2$ using the majorization order, we obtain four possibilities:
\begin{enumerate}
\item $\bm{\lambda}^\downarrow(\rho_{1}) = \bm{\lambda}^\downarrow(\rho_{2})$: They have the same coherence. 
\item $\bm{\lambda}^\downarrow(\rho_{1}) \precneqq \bm{\lambda}^\downarrow(\rho_{2})$: $\rho_{1}$ is less coherent than $\rho_{2}$. 
\item $\bm{\lambda}^\downarrow(\rho_{2}) \precneqq \bm{\lambda}^\downarrow(\rho_{1})$: $\rho_{1}$ is more coherent than $\rho_{2}$. 
\item $\bm{\lambda}^\downarrow(\rho_{1}) \parallel \bm{\lambda}^\downarrow(\rho_{2})$: Their coherence is incomparable. 
\end{enumerate}
As a sanity check, for any partially coherent wave $\rho$,
\begin{equation}
(\frac{1}{n},\frac{1}{n},\ldots,\frac{1}{n}) \precneqq \bm{\lambda}^\downarrow(\rho) \precneqq (1,0,\ldots,0).
\end{equation}
The majorization order provides a finer comparison than the entropy order. It can be shown that~\cite{marshall2011,cicalese2002}
\begin{equation}
\bm{\lambda}^{\downarrow}(\rho_{1}) \precneqq \bm{\lambda}^{\downarrow}(\rho_{2}) \implies H(\rho_{1}) > H(\rho_{2}),
\end{equation}
However, the converse is not necessarily true. In fact,
\begin{equation}
H(\rho_{1}) > H(\rho_{2}) \implies \bm{\lambda}^{\downarrow}(\rho_{1}) \precneqq \bm{\lambda}^{\downarrow}(\rho_{2}) \; \text{or} \; \bm{\lambda}^{\downarrow}(\rho_{1}) \parallel  \bm{\lambda}^{\downarrow}(\rho_{2}). 
\end{equation}
Consequently, the entropy order fails to capture the incomparable cases.

As examples that will be often used later, consider five $3\times 3$ density matrices $\rho_1$ to $\rho_5$ with
\begin{align}
\bm{\lambda}^\downarrow(\rho_1) &= (0.33,0.33,0.33),  \;
\bm{\lambda}^\downarrow(\rho_2) = (0.60,0.30,0.10), \notag \\
\bm{\lambda}^\downarrow(\rho_3) &= (0.80,0.15,0.05),  \;
\bm{\lambda}^\downarrow(\rho_4) = (1.00,0.00,0.00), \notag \\
\bm{\lambda}^\downarrow(\rho_5) &= (0.55,0.45,0.00). \label{eq:coherence_spectra_five_waves}
\end{align}
The majorization order indicates that
\begin{align}
&\bm{\lambda}^\downarrow(\rho_1) \precneqq \bm{\lambda}^\downarrow(\rho_2) \precneqq \bm{\lambda}^\downarrow(\rho_3) \precneqq \bm{\lambda}^\downarrow(\rho_4), \label{eq:relation_chain_1}\\
&\bm{\lambda}^\downarrow(\rho_1) \precneqq \bm{\lambda}^\downarrow(\rho_5) \precneqq \bm{\lambda}^\downarrow(\rho_4), \label{eq:relation_chain_2}\\
&\bm{\lambda}^\downarrow(\rho_2) \parallel \bm{\lambda}^\downarrow(\rho_5), \quad \bm{\lambda}^\downarrow(\rho_3) \parallel \bm{\lambda}^\downarrow(\rho_5), \label{eq:relation_chain_3}
\end{align}
In contrast, the entropy order indicates that
\begin{equation}
H(\rho_1) > H(\rho_2) > H(\rho_5) > H(\rho_3) > H(\rho_4).
\end{equation}
These relations are summarized in a Hasse diagram~\cite{davey2002} (Fig.~\ref{fig:geometry}a), where the edges indicate the majorization order and the height indicates the entropy order.

\begin{figure}[hbtp]
    \centering
    \includegraphics[width=0.45\textwidth]{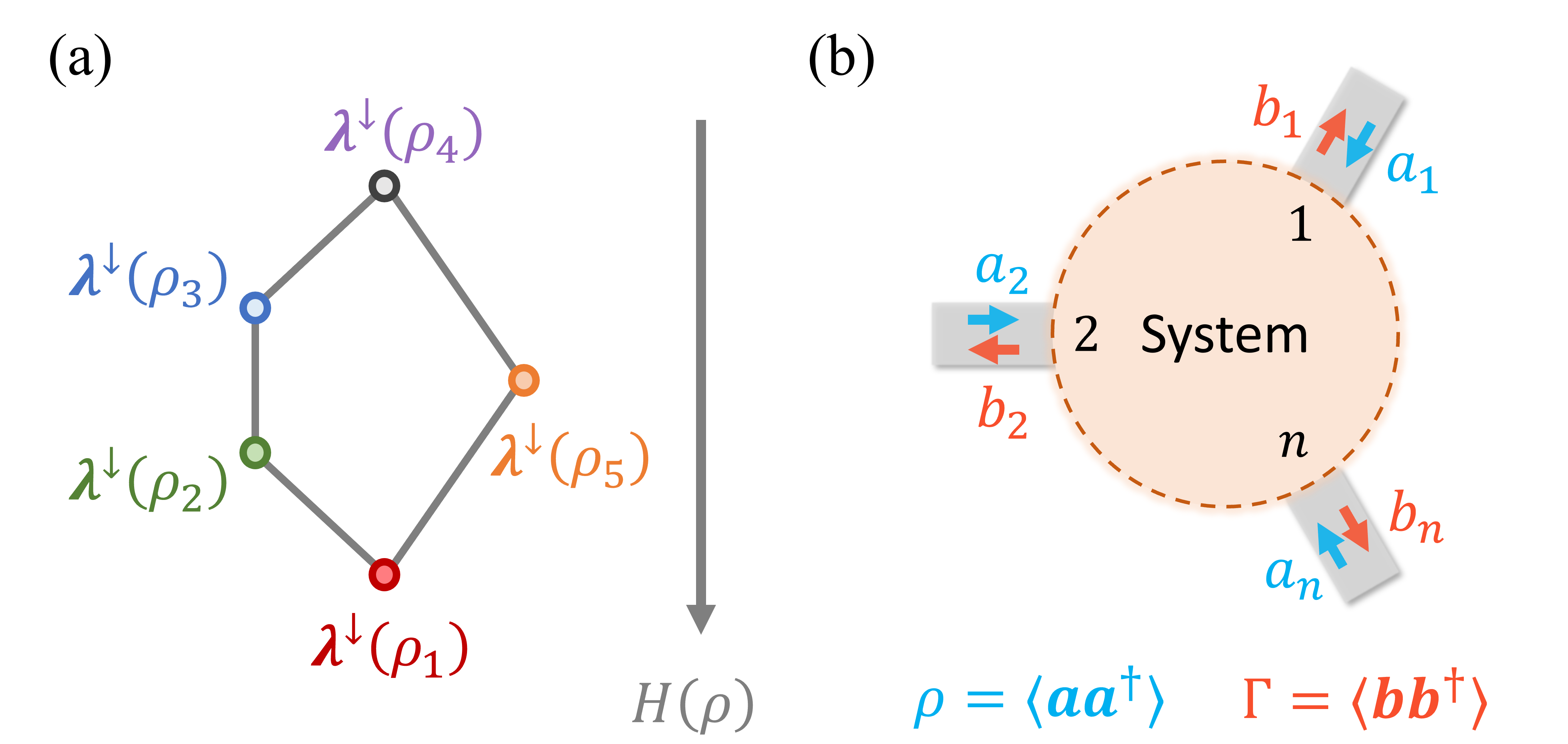}
    \caption{(a) Hasse diagram for $\bm{\lambda}^{\downarrow}\left(\rho_{1}\right)$ to $\bm{\lambda}^{\downarrow}\left(\rho_{5}\right)$. An edge indicates a strict majorization relation between the lower and upper vertices. A higher vertical position indicates a lower entropy. (b) An $n$-port linear time-invariant system.
    }
    \label{fig:geometry}
\end{figure}

This work aims to demonstrate the fundamental role of the majorization order in transport processes. Resource theories treat coherence as a resource that constrains achievable observables~\cite{streltsov2017,chitambar2019}. This perspective motivates us to examine the range of achievable transport responses for input waves with a specific coherence spectrum $\bm{\lambda}^\downarrow(\rho)$. We will show that waves with lower coherence in the majorization order exhibit more constrained ranges of achievable outcomes in transport processees.

Specifically, we consider a linear system where an input wave $\rho$ yields a response $F(\rho)$, with $F$ representing power distribution, absorption, transmission, or reflection. We generate all waves with identical total power and coherence spectrum as $\rho$ via unitary control, which transforms the input wave according to
\begin{equation}\label{eq:unitary_control}
\rho \to \rho[U] = U \rho U^{\dagger}.
\end{equation}
We examine the set of all achievable responses:
\begin{equation}
\{F\} := \{F(\rho[U]) \mid U \in U(n)\}.
\end{equation}
We show that this set preserves the majorization order: For sets $\{F\}_{1}$ and $\{F\}_{2}$ corresponding to waves $\rho_{1}$ and $\rho_{2}$, respectively,
\begin{equation}~\label{eq:general_F_result}
\bm{\lambda}^{\downarrow}(\rho_{1}) \prec \bm{\lambda}^{\downarrow}(\rho_{2}) \implies \{F\}_{1} \subseteq \{F\}_{2}.
\end{equation}
This result reveals the direct physical consequences of the majorization order. Moreover, the converse of Eq.~(\ref{eq:general_F_result}) often holds. Thus, measuring $\{F\}$ enables experimental comparison of coherence in the majorization order.

We begin our detailed analysis by reviewing the scattering matrix and unitary control. Consider an $n$-port linear time-invariant system characterized by a scattering matrix $S \in M_n$~\cite{haus1984} (Fig.~\ref{fig:geometry}b). A coherent input wave, represented by vector $\bm{a}$, scatters into an output wave $\bm{b} = S \bm{a}$. A partially coherent input wave is described by a density matrix:
\begin{equation}
\rho = \langle \bm{a} \bm{a}^{\dagger}\rangle,
\end{equation}
where $\langle \cdot \rangle$ denotes the ensemble average. The diagonal elements of $\rho$, denoted by $\bm{d}(\rho)$, specify the input power in each port, while $\operatorname{tr} \rho$ gives the total input power, assumed to be unity [Eq.~(\ref{eq:normalization_rho})]. The output wave is characterized by an unnormalized density matrix:
\begin{equation}\label{eq:Gamma}
\Gamma = \langle \bm{b} \bm{b}^{\dagger}\rangle = S \rho S^\dagger.
\end{equation}
The diagonal elements of $\Gamma$ represent the output power in each port, and $\operatorname{tr}\Gamma$ is the total output power, which may differ from unity in systems with loss or gain.

\begin{figure}[hbtp]
    \centering
    \includegraphics[width=0.45\textwidth]{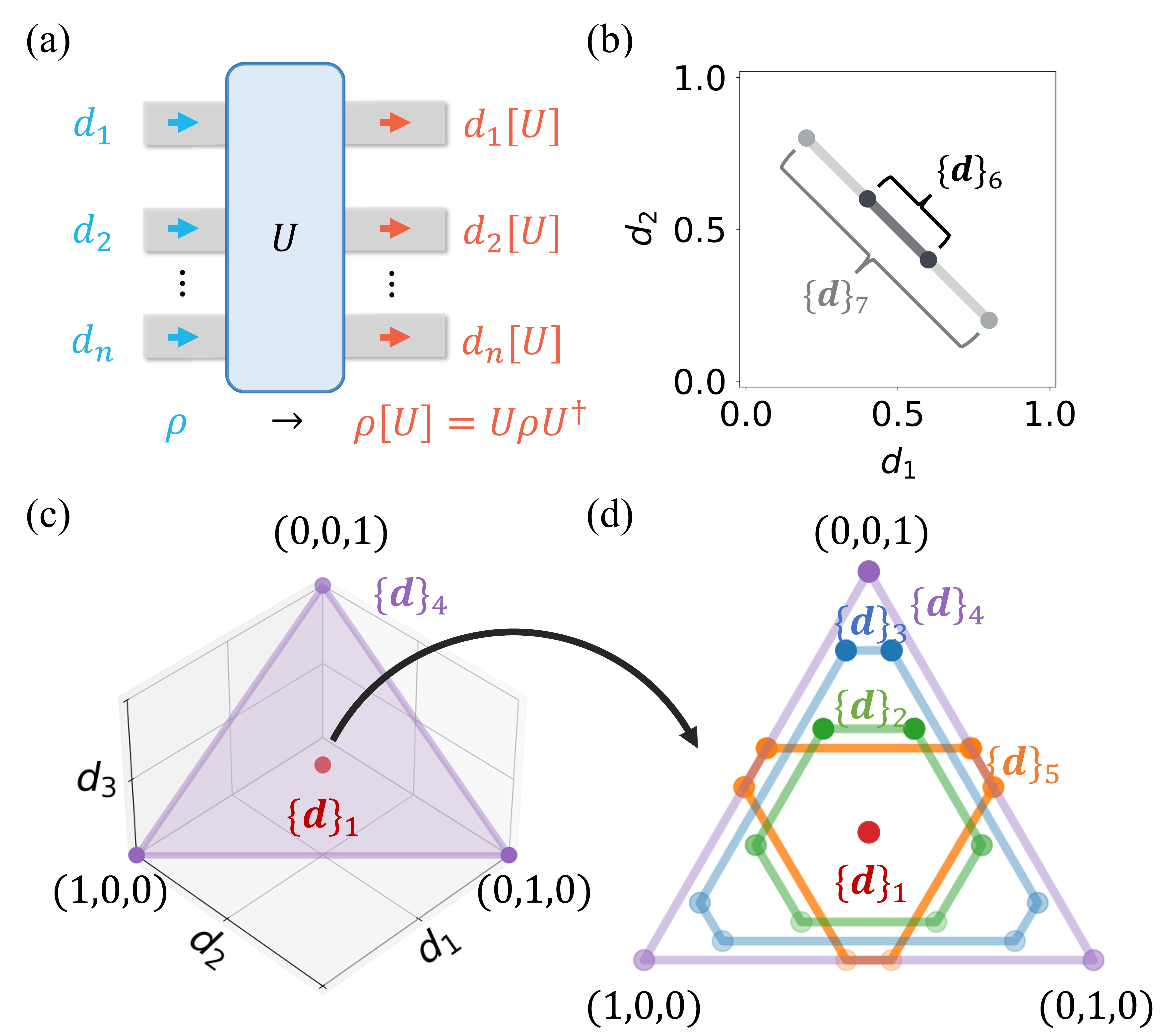}
    \caption{(a) Scheme of unitary control. The power distribution $\bm{d}(\rho) \to \bm{d}(U \rho U^\dagger)$. (b) $\{\bm{d}\}$ for $\rho_{6}$ and $\rho_{7}$. (c,d) $\{\bm{d}\}$ for $\rho_{1}$ to $\rho_{5}$. (c) A 3D plot. All $\{\bm{d}\}$ lie in the plane $d_1 + d_2 + d_3 =1$. (d) Each set $\{\bm{d}\}$ in the plane. Only the boundaries are shown.} 
    \label{fig:unitary_control}
\end{figure}

For a lossless system, the scattering matrix is unitary, denoted as $U$. The output density matrix then becomes $\rho[U]$ as defined in Eq.~(\ref{eq:unitary_control}). This process is called unitary control of the density matrix. Unitary control preserves both the total power and the coherence spectrum:
\begin{equation}
\operatorname{tr}\rho[U] = 1, \qquad \bm{\lambda}^\downarrow(\rho[U])=\bm{\lambda}^\downarrow(\rho).
\end{equation}
Conversely, any pair of waves with identical total power and coherence spectrum can be interconverted through unitary control. Therefore, the set
\begin{equation}
\{\rho[U] \mid U \in U(n)\}
\end{equation}
comprises all waves with the same total power and coherence spectrum as $\rho$.

Unitary control can be implemented using programmable unitary converters such as spatial light modulators~\cite{vellekoop2007,popoff2014,yu2017e}, Mach-Zehnder interferometers~\cite{reck1994,miller2013c,miller2013a,miller2013b,carolan2015,miller2015,clements2016,ribeiro2016,wilkes2016,annoni2017,miller2017d,perez2017,harris2018,pai2019}, and multiplane light conversion devices~\cite{morizur2010,labroille2014,tanomura2022,kupianskyi2023,taguchi2023,zhang2023b}. It has been introduced to manipulate the absorption, transmission, and reflection of both coherent~\cite{guo2023b,guo2024,guo2024c} and partially coherent waves~\cite{guo2024a,guo2024b}. Here, we examine four transport measurements under unitary control: power distribution, absorption, transmission, and reflection.

First, we consider the power distribution measurement (Fig.~\ref{fig:unitary_control}a). We apply unitary control [Eq.~(\ref{eq:unitary_control})] to an input wave $\rho$ and measure the power distribution in each port:
\begin{equation}
\bm{d}(\rho[U]) = \bm{d}(U \rho U^{\dagger}).
\end{equation}
It can be shown that the set of all achievable power distribution vectors under unitary control is:
\begin{align}
\{\bm{d}\} &\coloneqq \{\bm{d}(\rho[U]) \mid U\in U(n)\} \nonumber \\ &= \{\bm{u}\in \mathbb{R}^{n} \mid \bm{u} \prec \bm{\lambda}^{\downarrow}(\rho)\}.     
\label{eq:result_d_set}
\end{align}
Eq.~(\ref{eq:result_d_set}) can be proved using the Schur-Horn theorem~\cite{marshall2011,zhang2019a}. It has a simple geometric interpretation: $\{\bm{d}\}$ is the convex hull spanned by the $n!$ points obtained by permuting the coordinates of $\bm{\lambda}^\downarrow(\rho)$.

Next, consider two input waves $\rho_{1}$ and $\rho_{2}$ with their corresponding sets $\{\bm{d}\}_{1}$ and $\{\bm{d}\}_{2}$. One can prove that:
\begin{equation}\label{eq:result_summary_1}
\bm{\lambda}^{\downarrow}(\rho_{1})  \prec \bm{\lambda}^{\downarrow}(\rho_{2}) \iff \{\bm{d}\}_{1} \subseteq \{\bm{d}\}_{2}.
\end{equation}
More precisely, considering all four possibilities:
\begin{align}
 \bm{\lambda}^{\downarrow}(\rho_{1}) = \bm{\lambda}^{\downarrow}(\rho_{2}) &\iff \{\bm{d}\}_{1} = \{\bm{d}\}_{2} \label{eq:result_d_1}
 \\
 \bm{\lambda}^{\downarrow}(\rho_{1}) \precneqq \bm{\lambda}^{\downarrow}(\rho_{2}) &\iff \{\bm{d}\}_{1} \subsetneqq \{\bm{d}\}_{2} \label{eq:result_d_2}\\
\bm{\lambda}^{\downarrow}(\rho_{2}) \precneqq \bm{\lambda}^{\downarrow}(\rho_{1}) &\iff \{\bm{d}\}_{2} \subsetneqq \{\bm{d}\}_{1} \label{eq:result_d_3}\\
 \bm{\lambda}^{\downarrow}(\rho_{1}) \parallel \bm{\lambda}^{\downarrow}(\rho_{2}) &\iff \{\bm{d}\}_{1} \parallel \{\bm{d}\}_{2} \label{eq:result_d_4}
\end{align}
Here, $A\parallel B$ for two sets $A$ and $B$ means that they intersect, but neither is a subset of the other. See Appendix, Sec.~\ref{SI-sec:detailed_proof_power_distribution} for detailed proofs of Eqs.~(\ref{eq:result_summary_1})-(\ref{eq:result_d_4}). Therefore, the power distribution measurement exactly preserves the majorization order and offers an experimental method to probe the majorization order.

We illustrate these results with two examples.
Firstly, consider two $2\times 2$ density matrices $\rho_{6}$ and $\rho_{7}$ with
\begin{equation}
\bm{\lambda}^{\downarrow}(\rho_{6}) = (0.80,0.20), \quad \bm{\lambda}^{\downarrow}(\rho_{7}) = (0.60,0.40).
\end{equation}
Figure~\ref{fig:unitary_control}b depicts the sets $\{\bm{d}\}_{6}$ and $\{\bm{d}\}_{7}$.
These sets are line segments with endpoints obtained by permuting the coordinates of $\bm{\lambda}^{\downarrow}(\rho_{6})$ and $\bm{\lambda}^{\downarrow}(\rho_{7})$, respectively.
We note that $\{\bm{d}\}_{6} \subsetneqq \{\bm{d}\}_{7}$ because $\bm{\lambda}^{\downarrow}(\rho_{6}) \precneqq \bm{\lambda}^{\downarrow}(\rho_{7})$. 

Secondly, consider the five $3\times 3$ density matrices $\rho_{1}$ to $\rho_{5}$ introduced in Eq.~(\ref{eq:coherence_spectra_five_waves}).
Figure~\ref{fig:unitary_control}(c,d) depicts the sets $\{\bm{d}\}_{1}$ to $\{\bm{d}\}_{5}$.
These sets are convex hexagons with vertices obtained by permuting the coordinates of $\bm{\lambda}^{\downarrow}(\rho_{1})$ to $\bm{\lambda}^{\downarrow}(\rho_{5})$, respectively.
($\{\bm{d}\}_{1}$ and $\{\bm{d}\}_{4}$ are degenerate hexagons with coalescing vertices.)
We note that
\begin{align}
&\{\bm{d}\}_{1} \subsetneqq \{\bm{d}\}_{2} \subsetneqq \{\bm{d}\}_{3} \subsetneqq \{\bm{d}\}_{4}, \\
&\{\bm{d}\}_{1} \subsetneqq \{\bm{d}\}_{5} \subsetneqq \{\bm{d}\}_{4}, \\
&\{\bm{d}\}_{2} \parallel \{\bm{d}\}_{5}, \quad  \{\bm{d}\}_{3} \parallel \{\bm{d}\}_{5},
\end{align}
which confirm Eqs.~(\ref{eq:result_d_1})-(\ref{eq:result_d_4}) applied to Eqs.~(\ref{eq:relation_chain_1})-(\ref{eq:relation_chain_3}).

Second, we consider the absorption measurement (Fig.~\ref{fig:absorption_transmission}a). We input a wave $\rho$ into a system with a scattering matrix $S$ and measure the total absorption:
\begin{equation}
\alpha[S] = \operatorname{tr} (\rho A),
\end{equation}
where $A$ is the absorptivity matrix~\cite{miller2017b,guo2023b}, defined as
\begin{equation}
A \coloneqq I -S^{\dagger} S.
\end{equation}
We apply unitary control [Eq.~(\ref{eq:unitary_control})] to transform the input wave $\rho$. The total absorption then changes to
\begin{equation}
\alpha[S] \to \alpha[S,U] = \operatorname{tr}(U\rho U^{\dagger}A).
\end{equation}
It has been shown that the set of all achievable total absorption values under unitary control is~\cite{guo2024a}
\begin{align}
\{ \alpha \}_S &\coloneqq \{\alpha[S,U] \mid U\in U(n)\}  \nonumber \\ &= \left[\bm{\lambda}^{\downarrow}(\rho)\cdot \bm{\lambda}^{\uparrow}(A),\bm{\lambda}^{\downarrow}(\rho)\cdot \bm{\lambda}^{\downarrow}(A) \right]\label{eq:answer_absorption_interval}
\end{align}
with the absorption eigenvalues given by~\cite{yamilov2016}
\begin{equation}
\bm{\lambda}^{\downarrow}(A) = \bm{1}-\bm{\sigma}^{2 \uparrow}(S) = (1-\sigma^{2\uparrow}_1, \ldots, 1-\sigma^{2 \uparrow}_n),
\end{equation}
where $\bm{\sigma}(S)$ denotes the vector of singular values of $S$, $\uparrow$ indicates reordering the components in nondecreasing order, $[~,~]$ denotes the close real interval, and $\cdot$ represents the usual inner product.

Next, consider two waves $\rho_{1}$ and $\rho_{2}$ with their corresponding sets $\{\alpha\}_{S,1}$ and $\{\alpha\}_{S,2}$. One can prove that~\cite{guo2024a}:
\begin{equation}\label{eq:result_summary_2}
\bm{\lambda}^{\downarrow}(\rho_{1})  \prec \bm{\lambda}^{\downarrow}(\rho_{2}) \iff \forall S \in M_{n}, \{\alpha\}_{S,1} \subseteq \{\alpha\}_{S,2}.
\end{equation}
Thus, the total absorption measurement preserves the majorization order. More precisely,
\begin{align}
\bm{\lambda}^{\downarrow}(\rho_{1}) = \bm{\lambda}^{\downarrow}(\rho_{2}) \iff &\forall S \in M_{n}, \{\alpha\}_{S,1} = \{\alpha\}_{S,2} \label{eq:result_alpha_1}
\\
\bm{\lambda}^{\downarrow}(\rho_{1}) \precneqq \bm{\lambda}^{\downarrow}(\rho_{2}) \iff &\forall S \in M_{n}, \{\alpha\}_{S,1} \subseteq \{\alpha\}_{S,2} \And \notag \\ &\exists \tilde{S} \in M_{n}, \{\alpha\}_{\tilde{S},1} \neq \{\alpha\}_{\tilde{S},2} \label{eq:result_alpha_2}\\
\bm{\lambda}^{\downarrow}(\rho_{2}) \precneqq \bm{\lambda}^{\downarrow}(\rho_{1}) \iff &\forall S \in M_{n}, \{\alpha\}_{S,2} \subseteq \{\alpha\}_{S,1} \And  \notag \\ &\exists \tilde{S} \in M_{n}, \{\alpha\}_{\tilde{S},2} \neq  \{\alpha\}_{\tilde{S},1} \label{eq:result_alpha_3}\\
\bm{\lambda}^{\downarrow}(\rho_{1}) \parallel \bm{\lambda}^{\downarrow}(\rho_{2}) \iff &\exists S \in M_{n}, \{\alpha\}_{S,1} \nsubseteq \{\alpha\}_{S,2} \And \nonumber \\ &\exists \tilde{S} \in M_{n}, \{\alpha\}_{\tilde{S},2} \nsubseteq \{\alpha\}_{\tilde{S},1}\label{eq:result_alpha_4}
\end{align}
Here $\&$ denotes ``and". See Appendix, Sec.~\ref{SI-sec:detailed_proof} for detailed proofs of Eqs.~(\ref{eq:result_alpha_1})-(\ref{eq:result_alpha_4}). Note the differences between Eqs.~(\ref{eq:result_alpha_1})-(\ref{eq:result_alpha_4}) and Eqs.~(\ref{eq:result_d_1})-(\ref{eq:result_d_4}).

As an illustration, Fig.~\ref{fig:absorption_transmission}b depicts $\{\alpha\}_S$ for $\rho_1$ to $\rho_5$ in Eq.~(\ref{eq:coherence_spectra_five_waves}) and a $3\times 3$ $S$-matrix with $\bm{\sigma}^\downarrow(S)=(0.95,0.39,0.32)$. We note that  
\begin{align}
&\{\alpha\}_{S,1} \subsetneqq \{\alpha\}_{S,2} \subsetneqq \{\alpha\}_{S,3} \subsetneqq \{\alpha\}_{S,4}, \\
&\{\alpha\}_{S,1} \subsetneqq \{\alpha\}_{S,5} \subsetneqq \{\alpha\}_{S,4}, \\
&\{\alpha\}_{S,2} \parallel \{\alpha\}_{S,5}, \quad  \{\alpha\}_{S,3} \parallel \{\alpha\}_{S,5},
\end{align}
which confirm Eqs.~(\ref{eq:result_alpha_1})-(\ref{eq:result_alpha_4}) applied to Eqs.~(\ref{eq:relation_chain_1})-(\ref{eq:relation_chain_3}).

The total absorption measurement can also probe the majorization order. We consider two experimental settings. In the first setting, we have a single lossy system with an unknown scattering matrix $S$. We perform total absorption measurements under unitary control and obtain $\{\alpha\}_{S,1}$ and $\{\alpha\}_{S,2}$. By comparing $\{\alpha\}_{S,1}$ and $\{\alpha\}_{S,2}$, we can infer the relation between $\bm{\lambda}^{\downarrow}(\rho_{1})$ and $\bm{\lambda}^{\downarrow}(\rho_{2})$~\footnote{These are the only four possibilities since both $\{\alpha\}_{S,1}$ and $\{\alpha\}_{S,2}$ contain the point $\frac{1}{n}\sum_{i=1}^{n}[1-\sigma^{2}_{i}(S)]$}: \begin{align}
\{\alpha\}_{S,1} = \{\alpha\}_{S,2} \implies& \text{No information;}  \label{eq:single-measurement-1}\\
\{\alpha\}_{S,1} \subsetneqq \{\alpha\}_{S,2} \implies& \bm{\lambda}^{\downarrow}(\rho_{1}) \precneqq \bm{\lambda}^{\downarrow}(\rho_{2}) \;\text{or} \nonumber \\ &\bm{\lambda}^{\downarrow}(\rho_{1}) \parallel \bm{\lambda}^{\downarrow}(\rho_{2});  \label{eq:single-measurement-2}\\
\{\alpha\}_{S,2} \subsetneqq \{\alpha\}_{S,1} \implies& \bm{\lambda}^{\downarrow}(\rho_{2}) \precneqq \bm{\lambda}^{\downarrow}(\rho_{1}) \; \text{or} \nonumber \\  &\bm{\lambda}^{\downarrow}(\rho_{1}) \parallel \bm{\lambda}^{\downarrow}(\rho_{2});  \label{eq:single-measurement-3}\\
\{\alpha\}_{S,1} \parallel \{\alpha\}_{S,2} \implies& \bm{\lambda}^{\downarrow}(\rho_{1}) \parallel \bm{\lambda}^{\downarrow}(\rho_{2}). \label{eq:single-measurement-4}
\end{align} 
See Appendix, Sec.~\ref{SI-sec:measurement_proof} for detailed proofs of Eqs.~(\ref{eq:single-measurement-1})-(\ref{eq:single-measurement-4}). Note that only the last case yields a definite relation [Eq.~(\ref{eq:single-measurement-4})].

Eqs.~(\ref{eq:single-measurement-1})-(\ref{eq:single-measurement-4}) demonstrate that a single lossy system may not provide sufficient information to definitively determine the relation between arbitrary $\bm{\lambda}^\downarrow(\rho_{1})$ and $\bm{\lambda}^\downarrow(\rho_{2})$. To address this limitation, we perform absorption measurements on a set of lossy systems with carefully designed scattering matrices. The minimum number of systems required is $\left\lceil \frac{n-1}{2} \right\rceil$, where $\lceil \cdot \rceil$ represents the ceiling function. This number is necessary because comparing $\{\alpha\}_{S,1}$ and $\{\alpha\}_{S,2}$ in one system yields two inequalities, while verifying the majorization order requires $(n-1)$ inequalities [Eq.~(\ref{eq:def_majorization_ineq})]. To show this number is also sufficient, consider the following set of systems: \begin{equation} S_{m} = \begin{pmatrix} I_{n-m} & O \\ O & O \end{pmatrix}, \quad m=1,2,\dots,\left\lceil \frac{n-1}{2} \right\rceil. \end{equation} Comparing $\{\alpha\}_{S_{m},1}$ and $\{\alpha\}_{S_{m},2}$ enables verification of all $(n-1)$ linear inequalities required for majorization, thus providing sufficient information to determine the definitive relation between any $\bm{\lambda}^\downarrow(\rho_{1})$ and $\bm{\lambda}^\downarrow(\rho_{2})$. Further investigation is needed to develop a general criterion for assessing the sufficiency of an arbitrary set of systems.

\begin{figure}
    \centering
    \includegraphics[width=0.4\textwidth]{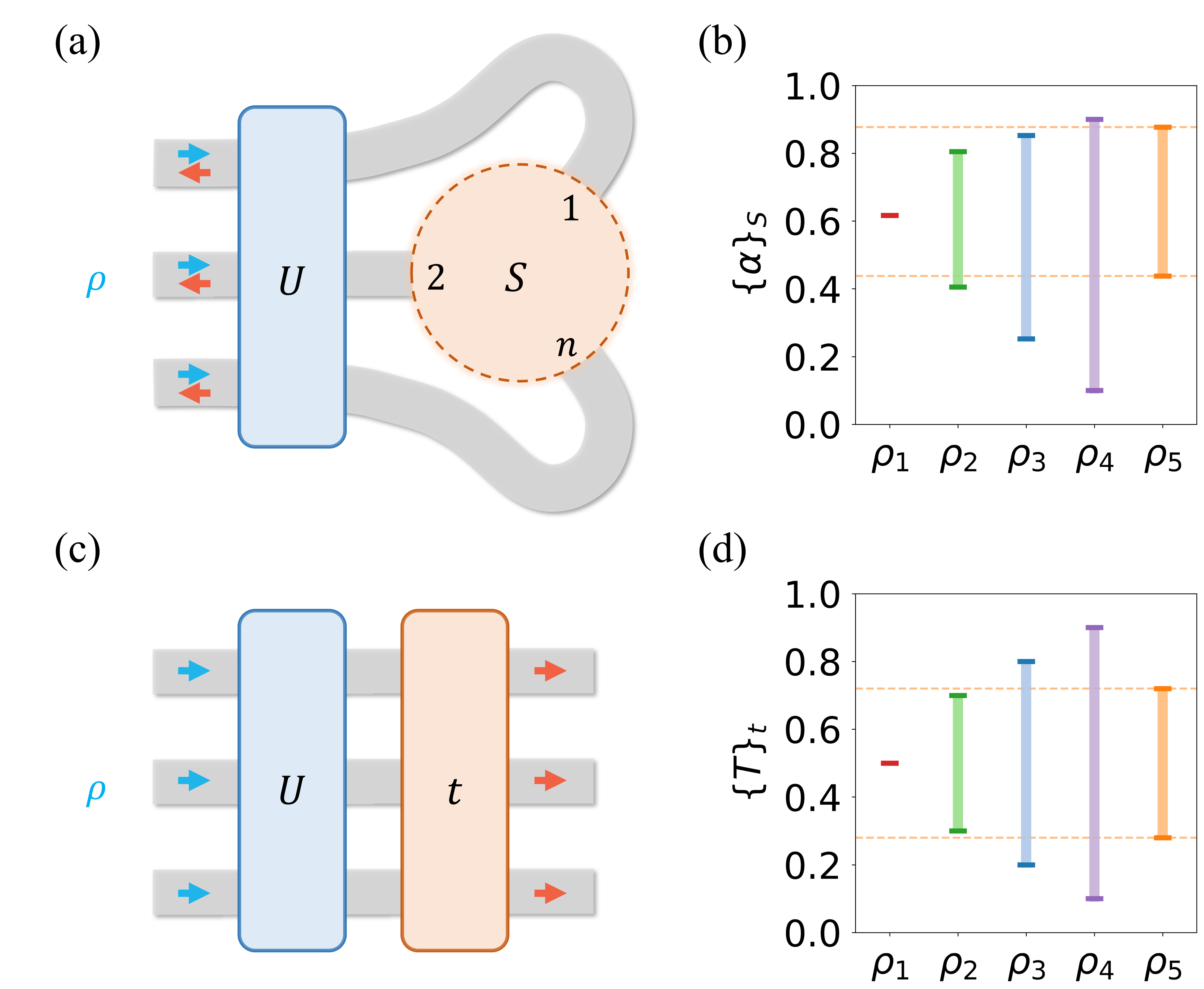}
    \caption{(a) Total absorption measurement with unitary control. (b) $\{\alpha\}_S$ for $\rho_1$ to $\rho_5$ where $\bm{\sigma}^\downarrow(S)=(0.95,0.39,0.32)$. (c) Total transmission measurement with unitary control. (d) $\{T\}_t$ for $\rho_1$ to $\rho_5$ where $\bm{\sigma}^\downarrow(t)=(0.95,0.71,0.32)$. }
    \label{fig:absorption_transmission}
\end{figure}

Third, we consider the total transmission measurement (Fig.~\ref{fig:absorption_transmission}c). We examine a system with a $2n \times 2n$ scattering matrix
\begin{equation}
    \tilde{S} = \begin{pmatrix}
        r & t' \\
        t & r'
    \end{pmatrix}, 
\end{equation}
where $r$ and $t$ are the $n\times n$ reflection and transmission matrices for input from the left, and $r'$ and $t'$ are the corresponding matrices for input from the right. We input a wave $\rho$ from the left and measure the total transmission:
\begin{equation}
T[t] = \operatorname{tr} (\rho t^{\dagger} t).
\end{equation}
We apply unitary control [Eq.~(\ref{eq:unitary_control})] to transform the input wave $\rho$. The total transmission then changes to
\begin{equation}
T[t] \to T[t,U] = \operatorname{tr} (U\rho U^{\dagger}t^{\dagger} t).
\end{equation}
It has been shown that the set of all achievable total transmission values under unitary control is~\cite{guo2024b}
\begin{align}
\{ T \}_t &\coloneqq \{T[t,U] \mid U\in U(n)\} \nonumber \\
&= \left[\bm{\lambda}^{\downarrow}(\rho)\cdot \bm{\sigma}^{2 \uparrow}(t),\bm{\lambda}^{\downarrow}(\rho)\cdot \bm{\sigma}^{2 \downarrow}(t)\right].
\end{align}

Next, consider two waves $\rho_{1}$ and $\rho_{2}$ with their corresponding sets $\{T\}_{t,1}$ and $\{T\}_{t,2}$. One can prove that~\cite{guo2024b}:
\begin{equation}
\bm{\lambda}^{\downarrow}(\rho_{1})  \prec \bm{\lambda}^{\downarrow}(\rho_{2}) \iff \forall t \in M_{n}, \{T\}_{t,1} \subseteq \{T\}_{t,2}.\label{eq:result_t_1}    
\end{equation}
The remaining discussion is analogous to that of absorption and is omitted. The analysis of reflection is similar. As an illustration, Fig.~\ref{fig:absorption_transmission}d depicts $\{T\}_t$ for $\rho_1$ to $\rho_5$ and a $3 \times 3$ $t$-matrix with $\bm{\sigma}^\downarrow(t)=(0.95,0.71,0.32)$, where
\begin{equation}
\{T\}_{t,1} \subsetneqq \{T\}_{t,2} \subsetneqq \{T\}_{t,5}  \subsetneqq \{T\}_{t,3} \subsetneqq \{T\}_{t,4},
\end{equation}
which confirms Eq.~(\ref{eq:result_t_1}) applied to Eqs.~(\ref{eq:relation_chain_1})-(\ref{eq:relation_chain_3}).

We make four final remarks. First, our findings apply to both classical and quantum waves, including optical, acoustic, and electronic varieties. Second, many of our results, especially those concerning incomparable cases, are not captured by other measures such as entropy order.  
Third, while we primarily compared coherence between wave pairs, this can be extended to multiple waves. The mathematical property that $\langle \Delta^{\downarrow}_{n}, \prec\rangle$ forms a complete lattice~\cite{davey2002,alberti1982,bapat1991} ensures that any subset of waves has a well-defined supremum and infimum (see Appendix, Sec.~\ref{SI-sec:majorization_lattice} and \ref{SI-Sec:multiple_waves}). Fourth, all discussed measurements require determining the range
of responses under unitary control. It suffices to find the unitary transformations that achieve the extremal responses, which can be solved using efficient variational algorithms~\cite{guerreschi2017,mitarai2018,schuld2019,mari2021,cerezo2021,bergholm2022},  without running over all unitary transformations $U(n)$. 

In conclusion, our investigation of the majorization order for comparing wave coherence has revealed its fundamental role in transport measurements. We have shown that these measurements preserve the majorization order under unitary control, enabling direct experimental probes of this order for wave coherence. Our work provides a unifying framework for understanding coherence phenomena in wave transport, paving the way for improved coherence characterization and engineering in both classical and quantum technologies.

\begin{acknowledgments}
C.G. thanks Dr.~Zhaoyou Wang for helpful discussions. This work is funded by a Simons Investigator in Physics
grant from the Simons Foundation (Grant No.~827065) and
by a Multidisciplinary University Research Initiative
(MURI) grant from the U.S. Air Force Office of Scientific Research (AFOSR) (Grant No.~FA9550-21-1-0312).
\end{acknowledgments}

\appendix

\section{Detailed proofs of Eqs.~(\ref{eq:result_summary_1})-(\ref{eq:result_d_4})}\label{SI-sec:detailed_proof_power_distribution}

In this section, we provide detailed proofs of Eqs.~(\ref{eq:result_summary_1})-(\ref{eq:result_d_4}) in the main text. From Eq.~(\ref{eq:result_d_set}), we obtain:
\begin{align}
\{\bm{d}\}_{1} &= \{\bm{u}\in \mathbb{R}^{n} \mid \bm{u} \prec \bm{\lambda}^{\downarrow}(\rho_{1})\}, \\
\{\bm{d}\}_{2} &= \{\bm{v}\in \mathbb{R}^{n} \mid \bm{v} \prec \bm{\lambda}^{\downarrow}(\rho_{2})\}.
\end{align}

\subsection{Proof of Eq.~(\ref{eq:result_summary_1})}
\begin{proof}
First, we prove the forward implication ``$\implies$". From $\bm{\lambda}^{\downarrow}(\rho_{1}) \prec \bm{\lambda}^{\downarrow}(\rho_{2})$ and the transitivity of majorization, we deduce that
\begin{align}
\bm{u} \in \{\bm{d}\}_{1} &\iff \bm{u} \prec \bm{\lambda}^{\downarrow}(\rho_{1}) \implies \nonumber \\
\bm{u} \prec \bm{\lambda}^{\downarrow}(\rho_{2}) &\iff  \bm{u} \in \{\bm{d}\}_{2}.    
\end{align}
Thus, we obtain
\begin{equation}
\bm{\lambda}^{\downarrow}(\rho_{1}) \prec \bm{\lambda}^{\downarrow}(\rho_{2}) \implies \{\bm{d}\}_{1} \subseteq \{\bm{d}\}_{2}.     
\end{equation}
Second, we prove the reverse implication ``$\impliedby$". From $\{\bm{d}\}_{1} \subseteq \{\bm{d}\}_{2}$ and the reflexivity of majorization, we deduce that
\begin{align}
\bm{\lambda}^{\downarrow}(\rho_{1}) \prec \bm{\lambda}^{\downarrow}(\rho_{1}) &\implies \bm{\lambda}^{\downarrow}(\rho_{1})  \in \{\bm{d}\}_{1} \implies \nonumber \\
\bm{\lambda}^{\downarrow}(\rho_{1})  \in \{\bm{d}\}_{2} &\implies \bm{\lambda}^{\downarrow}(\rho_{1})  \prec \bm{\lambda}^{\downarrow}(\rho_{2}).    
\end{align}
Thus, we obtain 
\begin{equation}
\bm{\lambda}^{\downarrow}(\rho_{1}) \prec \bm{\lambda}^{\downarrow}(\rho_{2}) \impliedby \{\bm{d}\}_{1} \subseteq \{\bm{d}\}_{2}.     
\end{equation}
This completes the proof of Eq.~(\ref{eq:result_summary_1}). 
\end{proof}

\subsection{Proof of Eq.~(\ref{eq:result_d_1})}
\begin{proof}
Note that
\begin{align}
\bm{\lambda}^{\downarrow}(\rho_{1}) = \bm{\lambda}^{\downarrow}(\rho_{2}) &\iff \bm{\lambda}^{\downarrow}(\rho_{1}) \prec \bm{\lambda}^{\downarrow}(\rho_{2}) \And \bm{\lambda}^{\downarrow}(\rho_{2}) \prec \bm{\lambda}^{\downarrow}(\rho_{1}) \nonumber \\
&\iff \{\bm{d}\}_{1} \subseteq \{\bm{d}\}_{2} \And \{\bm{d}\}_{2} \subseteq \{\bm{d}\}_{1}  \nonumber \\
&\iff \{\bm{d}\}_{1} = \{\bm{d}\}_{2}.
\end{align}
The first equivalence follows from the antisymmetry of $\prec$ as a partial order on $\Delta_{n}^{\downarrow}$. The second equivalence is due to Eq.~(\ref{eq:result_summary_1}). This completes the proof of Eq.~(\ref{eq:result_d_1}).
\end{proof}

\subsection{Proof of Eq.~(\ref{eq:result_d_2})}
\begin{proof}
Note that
\begin{align}
\bm{\lambda}^{\downarrow}(\rho_{1}) \precneqq \bm{\lambda}^{\downarrow}(\rho_{2}) &\iff \bm{\lambda}^{\downarrow}(\rho_{1}) \prec \bm{\lambda}^{\downarrow}(\rho_{2}) \And  \bm{\lambda}^{\downarrow}(\rho_{1}) \neq \bm{\lambda}^{\downarrow}(\rho_{2}) \nonumber \\
&\iff \{\bm{d}\}_{1} \subseteq \{\bm{d}\}_{2} \And \{\bm{d}\}_{1} \neq \{\bm{d}\}_{2}  \nonumber \\
&\iff \{\bm{d}\}_{1} \subsetneqq \{\bm{d}\}_{2}.
\end{align}
The first equivalence holds by definition. The second equivalence follows from Eq.~(\ref{eq:result_summary_1}) and the contrapositive of Eq.~(\ref{eq:result_d_1}). This completes the proof of Eq.~(\ref{eq:result_d_2}).
\end{proof}

\subsection{Proof of Eq.~(\ref{eq:result_d_3})}
\begin{proof}
The proof is similar to that of Eq.~(\ref{eq:result_d_2}), with the roles of $\rho_1$ and $\rho_2$ interchanged.
\end{proof}

\subsection{Proof of Eq.~(\ref{eq:result_d_4})}
\begin{proof}
First, we note that 
\begin{equation}
(\frac{1}{n},\frac{1}{n},\ldots,\frac{1}{n}) \prec \bm{\lambda}^{\downarrow}(\rho_{1}), \; (\frac{1}{n},\frac{1}{n},\ldots,\frac{1}{n}) \prec \bm{\lambda}^{\downarrow}(\rho_{2}).   
\end{equation}
It follows that
\begin{equation}\label{eq:d1_d2_intersect}
(\frac{1}{n},\frac{1}{n},\ldots,\frac{1}{n}) \in \{\bm{d}\}_{1}, \; (\frac{1}{n},\frac{1}{n},\ldots,\frac{1}{n}) \in \{\bm{d}\}_{2},    
\end{equation}
hence $\{\bm{d}\}_{1}$ and $\{\bm{d}\}_{2}$ intersects. Now we note that
\begin{align}
\bm{\lambda}^{\downarrow}(\rho_{1}) \parallel \bm{\lambda}^{\downarrow}(\rho_{2}) \iff &\neg [\bm{\lambda}^{\downarrow}(\rho_{1}) \prec \bm{\lambda}^{\downarrow}(\rho_{2})] \And  \nonumber \\ 
&\neg [\bm{\lambda}^{\downarrow}(\rho_{2}) \prec \bm{\lambda}^{\downarrow}(\rho_{1})]  \nonumber \\
\iff &\{\bm{d}\}_{1} \nsubseteq \{\bm{d}\}_{2} \And \{\bm{d}\}_{2} \nsubseteq \{\bm{d}\}_{1}  \nonumber\\
\iff &\{\bm{d}\}_{1} \parallel \{\bm{d}\}_{2},
\end{align} 
where $\neg$ denotes ``not". The first equivalence is by definition. The second equivalence follows from the contrapositive of Eq.~(\ref{eq:result_summary_1}). The third equivalence is due to Eq.~(\ref{eq:d1_d2_intersect}) and the definition of $\{\bm{d}\}_{1} \parallel \{\bm{d}\}_{2}$. This completes the proof of Eq.~(\ref{eq:result_d_4}).
\end{proof}

\section{Detailed proofs of Eqs.~(\ref{eq:result_summary_2})-(\ref{eq:result_alpha_4})}\label{SI-sec:detailed_proof}

In this section, we provide detailed proofs of Eqs.~(\ref{eq:result_summary_2})-(\ref{eq:result_alpha_4}) in the main text. From Eq.~(\ref{eq:answer_absorption_interval}), we obtain:
\begin{align}
\{\alpha\}_{S,1} &= [\bm{\lambda}^{\downarrow}(\rho_{1})\cdot \bm{\lambda}^{\uparrow}(A), \bm{\lambda}^{\downarrow}(\rho_{1})\cdot \bm{\lambda}^{\downarrow}(A)], \\
\{\alpha\}_{S,2} &= [\bm{\lambda}^{\downarrow}(\rho_{2})\cdot \bm{\lambda}^{\uparrow}(A), \bm{\lambda}^{\downarrow}(\rho_{2})\cdot \bm{\lambda}^{\downarrow}(A)].
\end{align}

\subsection{Proof of Eq.~(\ref{eq:result_summary_2})}
\begin{proof} 

The proof of the forward implication ``$\implies$" can be found in Ref.~\cite{guo2024a}, Eq.~(78). To prove the reverse implication ``$\impliedby$", consider the set of systems 
\begin{equation}
\tilde{S}_{k} = \begin{pmatrix}
I_{k} & O  \\
O & O
\end{pmatrix}, \qquad k=1,2,\dots n-1.    
\end{equation}
By examining the right endpoints of the absorption intervals, we obtain:
\begin{equation}~\label{eq:SI-proof_majorization_ineq}
\{\alpha\}_{\tilde{S}_{k},1} \subseteq \{\alpha\}_{\tilde{S}_{k},2} \implies \sum_{i=1}^{k} \lambda_{i}^{\downarrow}(\rho_{1}) \leq \sum_{i=1}^{k} \lambda_{i}^{\downarrow}(\rho_{2}).    
\end{equation}
Also, by the normalization condition~(\ref{eq:normalization_rho}),
\begin{equation}~\label{eq:SI-proof_majorization_eq}
\sum_{i=1}^{n} \lambda_{i}^{\downarrow}(\rho_{1}) = \sum_{i=1}^{n} \lambda_{i}^{\downarrow}(\rho_{2}) = 1.    
\end{equation}
Combining Eqs.~(\ref{eq:SI-proof_majorization_ineq}) and (\ref{eq:SI-proof_majorization_eq}), we obtain:
\begin{equation}
\bm{\lambda}^{\downarrow}(\rho_{1}) \prec \bm{\lambda}^{\downarrow}(\rho_{2}).    
\end{equation}
This completes the proof of Eq.~(\ref{eq:result_summary_2}).   
\end{proof}

\subsection{Proof of Eq.~(\ref{eq:result_alpha_1})}
\begin{proof}
Note that
\begin{align}
&\bm{\lambda}^{\downarrow}(\rho_{1}) = \bm{\lambda}^{\downarrow}(\rho_{2}) \iff \bm{\lambda}^{\downarrow}(\rho_{1}) \prec \bm{\lambda}^{\downarrow}(\rho_{2}) \And \bm{\lambda}^{\downarrow}(\rho_{2}) \prec \bm{\lambda}^{\downarrow}(\rho_{1})  \nonumber \\
&\iff \forall S \in M_{n}, \{\alpha\}_{S,1} \subseteq \{\alpha\}_{S,2} \And \{\alpha\}_{S,2} \subseteq \{\alpha\}_{S,1}  \nonumber\\
&\iff \forall S\in M_{n}, \{\alpha\}_{S,1} = \{\alpha\}_{S,2}.
\end{align}
The first equivalence follows from the antisymmetry of $\prec$ as a partial order on $\Delta_{n}^{\downarrow}$. The second equivalence is due to Eq.~(\ref{eq:result_summary_2}). This completes the proof of Eq.~(\ref{eq:result_alpha_1}). 
\end{proof}

\subsection{Proof of Eq.~(\ref{eq:result_alpha_2})}
\begin{proof}
Note that
\begin{align}
\bm{\lambda}^{\downarrow}(\rho_{1}) \precneqq \bm{\lambda}^{\downarrow}(\rho_{2}) &\iff \bm{\lambda}^{\downarrow}(\rho_{1}) \prec \bm{\lambda}^{\downarrow}(\rho_{2}) \And  \bm{\lambda}^{\downarrow}(\rho_{1}) \neq \bm{\lambda}^{\downarrow}(\rho_{2})  \nonumber \\
&\iff \forall S \in M_{n}, \{\alpha\}_{S,1} \subseteq \{\alpha\}_{S,2} \And \nonumber \\
&\qquad \quad \exists \tilde{S} \in M_{n}, \{\alpha\}_{\tilde{S},1} \neq \{\alpha\}_{\tilde{S},2} 
\end{align}
The first equivalence holds by definition. The second equivalence follows from Eq.~(\ref{eq:result_summary_2}) and the contrapositive of Eq.~(\ref{eq:result_alpha_1}). This completes the proof of Eq.~(\ref{eq:result_alpha_2}). 
\end{proof}

\subsection{Proof of Eq.~(\ref{eq:result_alpha_3})}
\begin{proof}
The proof is similar to that of Eq.~(\ref{eq:result_alpha_2}), with the roles of $\rho_1$ and $\rho_2$ interchanged.    
\end{proof}

\subsection{Proof of Eq.~(\ref{eq:result_alpha_4})}
\begin{proof}
Note that
\begin{align}
\bm{\lambda}^{\downarrow}(\rho_{1}) \parallel \bm{\lambda}^{\downarrow}(\rho_{2}) \iff &\neg [\bm{\lambda}^{\downarrow}(\rho_{1}) \prec \bm{\lambda}^{\downarrow}(\rho_{2})] \And  \nonumber \\ 
&\neg [\bm{\lambda}^{\downarrow}(\rho_{2}) \prec \bm{\lambda}^{\downarrow}(\rho_{1})]  \nonumber \\
\iff &\exists S \in M_{n}, \{\alpha\}_{S,1} \nsubseteq \{\alpha\}_{S,2} \And \nonumber \\ &\exists \tilde{S} \in M_{n}, \{\alpha\}_{\tilde{S},2} \nsubseteq \{\alpha\}_{\tilde{S},1}
\end{align} 
The first equivalence holds by definition. The second equivalence follows from the contrapositive of Eq.~(\ref{eq:result_summary_2}).
This completes the proof of Eq.~(\ref{eq:result_alpha_4}).
\end{proof}

\section{Detailed proofs of Eqs.~(\ref{eq:single-measurement-1})-(\ref{eq:single-measurement-4})}\label{SI-sec:measurement_proof}

In this section, we provide detailed proofs of Eqs.~(\ref{eq:single-measurement-1})-(\ref{eq:single-measurement-4}) in the main text. We prove them by considering all four possible relations between $\bm{\lambda}^{\downarrow}(\rho_{1})$ and $\bm{\lambda}^{\downarrow}(\rho_{2})$.

\subsection{Proof of Eq.~(\ref{eq:single-measurement-1})}
\begin{proof}
For example, consider the case when $S= \kappa I$ with $0 \leq \kappa \leq 1$. Then, for any $\rho_{1}$ and $\rho_{2}$,
\begin{equation}
\{\alpha\}_{S,1} = \{\alpha\}_{S,2} = \{1-\kappa^{2}\}.    
\end{equation}
Thus, in this case, we cannot infer any information about the relation between $\rho_{1}$ and $\rho_{2}$.    
\end{proof}

\subsection{Proof of Eq.~(\ref{eq:single-measurement-2})}
\begin{proof}
Since $\{\alpha\}_{S,1} \subsetneqq \{\alpha\}_{S,2}$, by Eqs.~(\ref{eq:result_alpha_1}) and (\ref{eq:result_alpha_3}), it is impossible that 
\begin{equation}
\bm{\lambda}^{\downarrow}(\rho_{2}) \precneqq \bm{\lambda}^{\downarrow}(\rho_{1}) \text{ or } \bm{\lambda}^{\downarrow}(\rho_{1}) = \bm{\lambda}^{\downarrow}(\rho_{2}).     
\end{equation}
The other two cases
\begin{equation}
\bm{\lambda}^{\downarrow}(\rho_{1}) \precneqq \bm{\lambda}^{\downarrow}(\rho_{2}) \text{ or } \bm{\lambda}^{\downarrow}(\rho_{1}) \parallel \bm{\lambda}^{\downarrow}(\rho_{2}),    
\end{equation}
are both possible. For example, consider $\bm{\lambda}(\rho_{3})$, $\bm{\lambda}(\rho_{4})$, and $\bm{\lambda}(\rho_{5})$ as defined in Eq.~(\ref{eq:coherence_spectra_five_waves}), which satisfy:
\begin{equation}\label{eq:relation_rho534}
\bm{\lambda}^{\downarrow}(\rho_{5}) \parallel \bm{\lambda}^{\downarrow}(\rho_{3}), \quad  \bm{\lambda}^{\downarrow}(\rho_{3}) \precneqq \bm{\lambda}^{\downarrow}(\rho_{4}).    
\end{equation}
Now, consider $S = \operatorname{diag}(0.95,0.39,0.32)$; thus, $\bm{\lambda}^{\downarrow}(A) = (0.95,0.71,0.31)$. From Eq.~(\ref{eq:answer_absorption_interval}), we obtain:
\begin{align}
\{\alpha\}_{S,3} &= [0.40,0.88], \\
 \{\alpha\}_{S,4} &= [0.31,0.95],  \\
\{\alpha\}_{S,5} &= [0.49,0.84].
\end{align}
Thus, 
\begin{equation}\label{eq:alpha_rho534}
\{\alpha\}_{S,5} \subsetneqq \{\alpha\}_{S,3} \subsetneqq \{\alpha\}_{S,4}.     
\end{equation}  
Eqs.~(\ref{eq:relation_rho534}) and (\ref{eq:alpha_rho534}) show that both cases are possible.
\end{proof}

\subsection{Proof of Eq.~(\ref{eq:single-measurement-3})}
\begin{proof}
The proof is similar to that of Eq.~(\ref{eq:single-measurement-2}), with the roles of $\rho_1$ and $\rho_2$ interchanged.    
\end{proof}
 
\subsection{Proof of Eq.~(\ref{eq:single-measurement-4})}
\begin{proof}
Note that
\begin{align}
\{\alpha\}_{S,1} \parallel \{\alpha\}_{S,2}  &\implies \{\alpha\}_{S,1} \nsubseteq \{\alpha\}_{S,2} \And \{\alpha\}_{S,2}\nsubseteq \{\alpha\}_{S,1}  \nonumber \\
&\implies \bm{\lambda}^{\downarrow}(\rho_{1}) \parallel \bm{\lambda}^{\downarrow}(\rho_{2}).
\end{align}
Here the first implication follows from the definition of $\parallel$, and the second implication follows from Eq.~(\ref{eq:result_alpha_4}).   
\end{proof}

\begin{remark}
We have proved that  
\begin{equation}
\exists S\in M_{n}, \{\alpha\}_{S,1} \parallel \{\alpha\}_{S,2}  \implies \bm{\lambda}^{\downarrow}(\rho_{1}) \parallel \bm{\lambda}^{\downarrow}(\rho_{2}).    
\end{equation}
However, the converse is not true, i.e., 
\begin{equation}
\bm{\lambda}^{\downarrow}(\rho_{1}) \parallel \bm{\lambda}^{\downarrow}(\rho_{2}) \centernot\implies \exists S\in M_{n}, \{\alpha\}_{S,1} \parallel \{\alpha\}_{S,2}.    
\end{equation}
For example, consider 
\begin{equation}
\bm{\lambda}^{\downarrow}(\rho_{1}) = \left( \frac{1}{2} , \frac{1}{4}, \frac{1}{4} , 0 \right), \;  \bm{\lambda}^{\downarrow}(\rho_{2}) = \left( \frac{3}{8} , \frac{1}{3} , \frac{1}{6} , \frac{1}{8} \right).     
\end{equation}
We denote 
\begin{align}
\bm{\mu} &= \bm{\lambda}^{\downarrow}(A) = \bm{1}- \bm{\sigma}^{2\uparrow}(S), \\
\bar{\mu} &= \frac{1}{4}(\mu_{1}+\mu_{2}+\mu_{3}+\mu_{4}), \\
r_1 &= \frac{1}{4}(\mu_{1}-\mu_{4}), \\
r_2 &= \frac{1}{8}(\mu_{1}-\mu_{4}) + \frac{1}{12}(\mu_{2}-\mu_{3}).
\end{align}
We calculate from Eq.~(\ref{eq:answer_absorption_interval}) that
\begin{align}
\{\alpha\}_{S,1} &= \left[ \bar{\mu} - r_1 , \bar{\mu} + r_1 \right],  \\ 
\{\alpha\}_{S,2} &= \left[ \bar{\mu} - r_2, \bar{\mu} + r_2 \right].
\end{align}
These two sets have three possible relationships:
\begin{equation}
\begin{cases}
\{\alpha\}_{S,1} \subsetneqq \{\alpha\}_{S,2},  & \text{when } r_1 < r_2; \\
\{\alpha\}_{S,1} = \{\alpha\}_{S,2}, & \text{when } r_1 = r_2; \\
\{\alpha\}_{S,2} \subsetneqq \{\alpha\}_{S,1}, & \text{when } r_1 > r_2.
\end{cases}   
\end{equation}
But it is impossible that
\begin{equation}
\{\alpha\}_{S,1} \parallel \{\alpha\}_{S,2}.    
\end{equation}
\end{remark}

\section{Ordered sets and lattices} \label{SI-sec:lattice}
In this section, we briefly review ordered sets and lattices. See Refs.~\cite{davey2002,roman2008} for a more detailed introduction.
\begin{definition}[partial order]
A \emph{partial order} on a nonempty set $P$ is a binary relation $\leqslant$ that is reflexive, antisymmetric, and transitive. Specifically, for all $x,y,z \in P$, the following properties hold:
\begin{enumerate}
\item Reflexivity: $x \leqslant x$.
\item Antisymmetry: $x \leqslant y$ and $y \leqslant x$ imply $x=y$.
\item Transitivity: $x \leqslant y$ and $y \leqslant z$ imply $x \leqslant z$.
\end{enumerate}
The pair $\langle P, \leqslant \rangle$ is called a \emph{partially ordered set} or \emph{poset}. If $x \leqslant y$ and $x \neq y$, we denote it as $x < y$. If $x \leqslant y$ or $y \leqslant x$, then $x$ and $y$ are said to be \emph{comparable}. Otherwise, they are \emph{incomparable}, denoted by $x \parallel y$.
\end{definition}
For subsets $S, T \subseteq P$, $S \leqslant T$ means $s \leqslant t$ for all $s \in S$ and $t \in T$. If $T = \{t\}$, we write $S \leqslant t$. Similarly, if $S = \{s\}$, we write $s \leqslant T$.
\begin{definition}[top and bottom]
A poset $\langle P, \leqslant \rangle$ is said to have
\begin{enumerate}
\item a top element $\top \in P$ if $x \leqslant \top$ for all $x \in P$;
\item a bottom element $\bot \in P$ if $\bot \leqslant x$ for all $x \in P$.
\end{enumerate}
A poset is called \emph{bounded} if it has both top and bottom elements.
\end{definition}
\begin{definition}[supremum and infimum]
Let $\langle P, \leqslant \rangle$ be a poset and let $S$ be a subset of $P$.
\begin{enumerate}
\item An \emph{upper bound} for $S$ is an element $x \in P$ such that $s \leqslant x$ for all $s \in S$. The set of all upper bounds for $S$ is denoted by $S^u$. If $S^u$ has a least element, it is called the \emph{supremum} or \emph{join} of $S$ and is denoted by $\bigvee S$. For a finite set $S = \{ a_1, \ldots, a_n \}$, the join is denoted by $a_1 \vee \ldots \vee a_n$. If $\bigvee S \in S$, it is called the maximum of $S$ and is denoted by $\max S$.
\item A \emph{lower bound} for $S$ is an element $x \in P$ such that $x \leqslant s$ for all $s \in S$. The set of all lower bounds for $S$ is denoted by $S^l$. If $S^l$ has a greatest element, it is called the \emph{infimum} or \emph{meet} of $S$ and is denoted by $\bigwedge S$. For a finite set $S = \{ a_1, \ldots, a_n \}$, the meet is denoted by $a_1 \wedge \ldots \wedge a_n$. If $\bigwedge S \in S$, it is called the minimum of $S$ and is denoted by $\min S$.
\end{enumerate}
\end{definition}
\begin{definition} [lattice and complete lattice]
Let $\langle P, \leqslant \rangle$ be a poset. Then:
\begin{enumerate}
\item $P$ is called a \emph{lattice} if every pair of elements in $P$ has a meet and a join.
\item $P$ is called a \emph{complete lattice} if every subset of $P$ has a meet and a join.
\end{enumerate}
\end{definition}
\begin{proposition}
Every complete lattice is bounded.
\end{proposition}

\section{Majorization lattice}\label{SI-sec:majorization_lattice}
In this section, we briefly review the mathematical fact that $\langle \Delta^{\downarrow}_{n}, \prec\rangle$ forms a complete lattice.
\begin{proposition}\label{prop:delta_n_poset}
The pair $\langle \Delta_n^\downarrow, \prec \rangle$ is a poset.
\end{proposition}
\begin{proof}
One can directly verify that the majorization relation is reflexive, antisymmetric, and transitive on $\Delta_n^\downarrow$.
\end{proof}
The poset $\langle \Delta_n^\downarrow, \prec \rangle$ possesses additional properties. For any $\bm{x}, \bm{y} \in \Delta_n^\downarrow$, their infimum $\bm{x} \wedge \bm{y}$ and supremum $\bm{x} \vee \bm{y}$ exist and belong to $\Delta_n^\downarrow$. Thus, the algebraic structure $\langle \Delta_n^\downarrow, \prec, \wedge, \vee \rangle$ forms a lattice, known as the \emph{majorization lattice}. Furthermore, this lattice is complete.
\begin{theorem}[Completeness of Majorization Lattice]\label{thm:majorization_lattice}
The algebraic structure $\langle \Delta_n^\downarrow, \prec, \wedge, \vee \rangle$ is a complete lattice. That is, for any subset $S \subseteq \Delta_n^\downarrow$, both the infimum $\bigwedge S$ and the supremum $\bigvee S$ exist and belong to $\Delta_n^\downarrow$.
\end{theorem}
\begin{proof}
See Ref.~\cite{alberti1982}, p.~55, Ref.~\cite{bapat1991}, Lemma~3, and Ref.~\cite{bosyk2019}, Proposition~1, for detailed proofs.
\end{proof}
\begin{corollary}\label{cor:majorization_lattice_bounded}
The majorization lattice is bounded, with $(1,0,\ldots,0)$ as the top element and $(\frac{1}{n},\frac{1}{n},\ldots,\frac{1}{n})$ as the bottom element. That is, for all $\bm{x} \in \Delta_n^\downarrow$,
\begin{equation}
(\frac{1}{n},\frac{1}{n},\ldots,\frac{1}{n}) \prec \bm{x} \prec (1,0,\ldots,0).
\end{equation}
\end{corollary}

\section{Majorization lattice for multiple waves}\label{SI-Sec:multiple_waves}

In the main text, we have applied the majorization order to compare the coherence between pairs of waves. In this section, we extend this comparison to multiple waves. Consider a set of waves 
\begin{equation}
\{\rho_{k}\} = \{\rho_1,\rho_2,\rho_3,\ldots\}.
\end{equation}
We aim to identify the set of waves less coherent than all $\rho_{k}$, denoted as $\{\rho_{l}\}$, and the set of waves more coherent than all $\rho_{k}$, denoted as $\{\rho_{u}\}$. The answer is: 
\begin{align}
\{\rho_{l}\} &= \{\rho : \bm{\lambda}^{\downarrow}(\rho) \precneqq \bm{\lambda}^{\downarrow}_{\inf}\}, \\
\{\rho_{u}\} &= \{\rho : \bm{\lambda}^{\downarrow}_{\sup} \precneqq \bm{\lambda}^{\downarrow}(\rho) \}.    
\end{align}
where $\bm{\lambda}^{\downarrow}_{\sup}$ and $\bm{\lambda}^\downarrow_{\inf}$ are the \emph{supremum} and \emph{infimum} of $\{\bm{\lambda}^\downarrow(\rho_{k})\}$ with respect to majorization, respectively; both values can be determined using known algorithms~\cite{bosyk2019}. Each of these sets is completely determined by a single strict majorization condition because the partially ordered set $\langle \Delta^{\downarrow}_{n}, \prec \rangle$ is a \emph{complete lattice}. 

The existence of supremum and infimum is invaluable in deriving fundamental bounds on physical responses for various coherence levels. For example, consider five waves $\rho_{l}$, $\rho_{\inf}$, $\rho_{k}$, $\rho_{\sup}$, and $\rho_{u}$ as indicated above. Then, 
\begin{equation}
\bm{\lambda}^{\downarrow}(\rho_{l}) \precneqq \bm{\lambda}^{\downarrow}_{\inf} \prec \bm{\lambda}^{\downarrow}(\rho_{k}) \prec \bm{\lambda}^{\downarrow}_{\sup} \precneqq \bm{\lambda}^{\downarrow}(\rho_{u}).    
\end{equation}
In the power distribution measurement,
\begin{equation}
\{ \bm{d} \}_{l} \subsetneqq \{ \bm{d} \}_{\inf} \subseteq \{ \bm{d} \}_{k} \subseteq \{ \bm{d} \}_{\sup} \subsetneqq \{ \bm{d} \}_{u}.    
\end{equation}
In the absorption measurement, for any $S\in M_{n}$,
\begin{equation}
\{ \alpha \}_{S,l} \subseteq \{ \alpha\}_{S,\inf} \subseteq \{ \alpha \}_{S,k} \subseteq \{ \alpha \}_{S,\sup} \subseteq \{ \alpha \}_{S,u}.    
\end{equation}
In the transmission measurement, for any $t\in M_{n}$,
\begin{equation}
\{ T \}_{t,l} \subseteq \{ T \}_{t,\inf} \subseteq \{ T \}_{t,k} \subseteq \{ T \}_{t,\sup} \subseteq \{ T \}_{t,u}.    
\end{equation}


\bibliography{main}

\end{document}